\documentclass[a4paper]{article}

\usepackage{fullpage}
\usepackage[utf8]{inputenc}
\usepackage{amsfonts}
\usepackage{amssymb}
\usepackage{pifont}

\usepackage{graphicx}
\usepackage{color}

\usepackage{url}
\usepackage[bookmarks,%
  colorlinks,%
  urlcolor=black,%
  citecolor=black,%
  linkcolor=black,%
  hyperfigures,%
  pagebackref,%
  pdfcreator=LaTeX,%
  breaklinks=true,%
  pdfpagelayout=SinglePage,%
  bookmarksopen=true,%
  bookmarksopenlevel=2]{hyperref}

\usepackage{amsthm}
\newtheorem{definition}{Definition}
\newtheorem{lemma}{Lemma}

\usepackage{tikz}

\usepackage[vlined,linesnumbered]{algorithm2e}
\SetKwBlock{SubAlgoBlock}{}{end}
\newcommand{\SubAlgo}[2]{#1 \SubAlgoBlock{#2}}
\let\oldnl\nl
\newcommand{\nonl}{\renewcommand{\nl}{\let\nl\oldnl}}

\newcommand{\concat}{\oplus}

\SetKwRepeat{DoWhile}{do}{while}

\SetKw{Operation}{operation}
\SetKw{Function}{function}
\SetKw{Is}{is:}
\SetKw{True}{true}
\SetKw{False}{false}

\SetKw{WFListType}{WFList}
\SetKw{InvokType}{Invocation}
\SetKw{ConsensusType}{Consensus}
\SetKw{InvokNodeType}{Inode}
\SetKw{HelpNodeType}{HNode}
\SetKw{StateType}{State}
\SetKw{ResponseType}{Response}
\SetKw{List}{List}
\SetKw{Type}{type}
\SetKw{ListType}{List}
\SetKw{TType}{T}
\SetKw{EmptyType}{Empty}
\SetKw{Queue}{queue}

\SetKwData{Invok}{invoc}
\SetKwData{Winner}{winner}
\SetKwData{Consensus}{cons}
\SetKwData{Announcements}{announce}
\SetKwData{Operations}{operations}
\SetKwData{InitialState}{initialState}

\SetKwData{ToHelp}{toHelp}
\SetKwData{State}{state}
\SetKwData{HNode}{hn}
\SetKwData{Result}{result}
\SetKwData{Res}{res}
\SetKwData{Inv}{inv}
\SetKwData{ToH}{toH}

\SetKwData{Log}{log}
\SetKwData{Decided}{decided}
\SetKwData{First}{first}
\SetKwData{Next}{next}
\SetKwData{Last}{last}
\SetKwData{Empty}{empty}
\SetKwData{Head}{head}
\SetKwData{Tail}{tail}
\SetKwData{Root}{root}
\SetKwData{List}{list}
\SetKwData{Node}{node}

\SetKwData{Cons}{cons}

\SetKwFunction{Enq}{enqueue}
\SetKwFunction{AppendFirst}{appendAtFirst}
\SetKwFunction{Append}{append}
\SetKwFunction{Lup}{lookup}
\SetKwFunction{Rem}{remove}
\SetKwFunction{NewConsensus}{consensus}
\SetKwFunction{Apply}{apply}
\SetKwFunction{Propose}{propose}
\SetKwFunction{Write}{write}
\SetKwFunction{Read}{read}
\SetKwFunction{Invoke}{invoke}
\SetKwFunction{Insert}{insert}
\SetKwFunction{Collect}{collect}
\SetKwFunction{CAux}{aux}
\SetKwFunction{CAS}{CAS}

\SetKw{Pair}{pair}
\SetKwData{Prepare}{prepare}
\SetKwData{Accept}{accept}
\SetKwData{Decide}{decide}
\SetKwFunction{Maxf}{max}
\SetKwData{Atom}{atomically}
\SetKwFunction{Firstf}{fst}
\SetKwFunction{Secondf}{snd}

\title{Wait-Free Universality of Consensus in the Infinite Arrival Model}

\author{Grégoire Bonin\\LS2N, Université de Nantes, France\\\url{gregoire.bonin@etu.univ-nantes.fr}
  \and Achour Mostéfaoui\\LS2N, Université de Nantes, France\\\url{achour.mostefaoui@univ-nantes.fr}
  \and Matthieu Perrin\\LS2N, Université de Nantes, France\\\url{matthieu.perrin@univ-nantes.fr}
}
\date{}

\begin{document}

\maketitle

\begin{abstract}
  In classical asynchronous distributed systems composed of a fixed number
  $n$ of processes where some proportion may fail by crashing, 
  many objects do not have a wait-free linearizable implementation (e.g. stacks,
  queues, etc.). It has been proved that consensus is universal in such systems,
  which means that this system
  augmented with consensus objects allows to implement any object that has
  a sequential specification. To this end, many universal constructions
  have been proposed in systems augmented with consensus objects or with
  different equivalent objects or special hardware instructions (compare\&swap, fetch\&add,
  etc.). In this paper, we consider a more general system model called
  infinite arrival model where infinitely many processes may arrive and
  leave or crash during a run. We prove that consensus is still universal
  in this more general model. For that, we propose a universal
  construction. As a first step we build a weak log for which we propose
  two implementations using consensus objects for the first and the
  compare\&swap special instruction for the other.

\end{abstract}

\section{Introduction}

The question of universality has always been central in all areas. What can be done with a given tool and a context, and more importantly what cannot be done with such tool. In sequential computing, universality is represented by a Turing machine that can compute all that is computable. In the context of distributed systems, we have known since 1985 and the famous FLP impossibility result that the consensus problem has no deterministic solution in a distributed system where even one process can fail by crashing \cite{FLP85}. This impossibility is not due to the computing power of the individual processes, but rather to the difficulty of coordination between the different processes that compose the distributed system. Coordination and agreement problems are thus at the heart of computability in distributed systems \cite{HRR13}.

A distributed system can be abstracted as a set of processes accessing concurrently a set of concurrent objects. The implementation of these objects are based on read/write registers and hardware instructions. Searching for correct and efficient implementations of usual objects (queues, stacks, etc.) is far from being trivial \cite{HS08, R13, T18} when the system is failure prone. Intuitively, a ``good'' implementation of a concurrent object has to satisfy two kinds of properties: a consistency condition and a progress condition.  The consistency condition specifies the safety property that is the meaningfulness of the returned results, and the progress condition specifies the guarantees on the liveness.

Linearizability  \cite{HW90} is a consistency criterion. It ensures that all the operations in a distributed history appear as if they were executed sequentially: each operation appears at a single point in time, between its start event and its end event. Such a consistency criterion gives the illusion to the processes to access a physical concurrent object. However, such implementations are often costly, when not impossible. 

\begin{definition}[Linearizability]
    An execution $\alpha$ is \emph{linearizable} if 
    all operations return the same value as if they occurred instantly at some instant, 
    called the \emph{linearization point}, between their invocation and their response, 
    possibly after removing some non-terminated operations. 
\end{definition}

The use of locks in the implementation may cause blocking in a system where processes can crash. Prohibiting the use of locks leads to several progress conditions, among which wait-freedom \cite{herlihy1991wait} and lock-freedom \cite{HW90}. While wait-freedom guarantees that every operation terminates after a finite time, lock-freedom guarantees that, if the computation run for long enough, at least one process makes progress (this may lead some processes to starve). Wait-freedom is thus stronger than lock-freedom: while lock-freedom is a system-wide condition, wait-freedom is a per-process condition.

\begin{definition}[Wait-freedom]
  An execution $\alpha$ is \emph{wait-free} if no operation takes an infinite number of
  steps in $\alpha$.
\end{definition}

Maurice Herlihy proved in \cite{H88} that consensus is universal in classical distributed systems composed of a set of $n$ processes. Namely, any object having a sequential specification has a wait-free implementation using only read/write registers (memory locations) and some number of consensus objects. 
 
For proving the universality of consensus, Herlihy introduced the notion of universal construction\footnote{A small guided tour on universal constructions can be found in \cite{R17}.}. It is a generic algorithm that, given a sequential specification of any object whose operations are total\footnote{This means that any operation on the object can be called and the call returns regardless of the state of the object.}, provides a concurrent implementation of this object. Since then, many universal constructions have been proposed for several objects, assuming the availability of hardware special instructions that provide the same computing power as consensus, like compare\&swap (CAS), Load-Link/Store-Conditional (LL/SC) etc.

This last decade, first with peer-to-peer systems, and then with multi-core machines and the multi-threading model, the assumption of a closed system with a fixed number $n$ of processes and where every process knows the identifiers of all processes became too restrictive. Hence the infinite arrival model introduced in \cite{merritt2003resilient}. In this model, any number of processes can crash (or leave, in a same way as in the other model), but any number (be it finite or not) of processes can also join the network. When a process joins such a system, it is not known to the already running processes, so no fixed number of processes can be used in the implementations as a parameter. Three kinds of arrival models can be distinguished: the bounded arrival model in which at most $N$ processes may participate, the unbounded arrival model in which a finite number of processes participate in each execution but the number of participants is unknown to them, and the infinite arrival model where new processes may keep arriving during the whole execution. Let us note that, at any time, the number of processes that have already joined the system is finite, but can be infinitely growing.
The classical system is part of the bounded arrival model where all the processes arrive at once.
In this article, we focus on the infinite arrival model. % as the particular case of 0-arrival model.

\subparagraph{Problem statement} 

The aim of this paper is to extend universality of consensus to the infinite arrival model. Solutions to the consensus problem have already been investigated for the infinite arrival model~\cite{aspnes2002wait, chockler2005active, merritt2003resilient} and consensus has been used as a base for reasoning about computability in this model~\cite{afek2011bounded}.
 The question is thus ``is it possible to build a universal wait-free linearizable construction based on consensus objects and read/write atomic registers?'' This is not trivial for different reasons. First, although the lock-free universal constructions still work in the infinite arrival model because they ensure a global progress condition, this is no more the case for wait-free universal constructions.
Second, wait-free implementations rely on what is called help mechanism, that has been recently formalized in \cite{censor2015help}. This mechanism requires any process, before terminating its operation, to help processes having pending operations, in order to reach wait-freedom. One of the difficulties in the infinite arrival model is that helping is not obvious. Indeed, helping requires at least that a process needing to be helped is able to announce its existence to other processes willing to help it. 
Due to the infinite number of potential participating processes over time, it is not reasonable to assume that each process can write in a dedicated register, and to require helping processes to read them all.
When only consensus and read/write registers are accessible to a process, a newly arriving process must compete with a potentially infinite number of other
arriving processes on either a consensus object or a same memory location; and may fail on all its attempts.

\subparagraph{Contributions}
This paper has two main contributions.

\begin{itemize}
\item Similarly to \cite{FK14} which first proposes a Collect object that will be used as a building block for a universal construction, we first propose a construction that implements a weak log object. This log is used as a list of presence where a process that arrives registers. We propose two implementations of a weak log using respectively consensus objects and the special hardware instructions compare\&swap.
\item We propose a universal construction based on consensus objects and the weak log object in the infinite arrival model where, moreover, processes are anonymous. This  proves that consensus is universal even in this model.
\end{itemize}

\subparagraph{Organization of the paper}

The remainder of this paper is organized as follows.
We first present the infinite arrival model in Section~\ref{section:system}.
Then, Section~\ref{section:weakLog} introduces a new abstraction, called
the weak log. In Section~\ref{section:universal}, we show how the weak log can
be used together with consensus to implement a wait-free universal construction.
In Section~\ref{section:consensus}, we give an implementation of the
weak log using consensus objects.
In Section~\ref{section:CAS}, we discuss a simpler algorithm that
uses compare\&swap instead of consensus to implement a weak log.
Finally, Section~\ref{section:conclusion} concludes the paper.

\section{System Model}
\label{section:system}

The infinite arrival model is composed of an infinite set of processes called $p_0, p_1, \dots$
Processes communicate through an infinite memory composed of unbounded 
registers \footnote{Memory addresses of an infinite memory are unbounded, so this assumption is necessary to store references.}.
Processes can only access atomic memory locations through references or thanks to a memory allocation mechanism that creates a finite 
number of memory locations when it is invoked. In other words, it is not possible to create an infinite array in this model. 
Memory is composed of three kinds of registers: immutable registers, read/write registers and consensus objects defined thereafter. 
\begin{description}
\item[An immutable register] $x$ is initialized with a fixed value and cannot be modified later. It provides a read operation, 
  simply denoted by $x$ in the algorithms, that returns the internal value of the register.
\item[A read/write register] $x$ provides two operations: a read operation $x.\Read()$
  that returns the internal value of the register, and a write operation $x.\Write(v)$ that 
  replaces the current value of $x$ by $v$.
\item[A consensus object] $x$ provides two operations: a read operation similar to the read operation on an immutable register, and
    an operation $x.\Propose(v)$ that atomically checks if the current value is
  $\bot$, then sets the value of $x$ to $v$ if it is the case, and finally returns the value of $x$. In other words, 
  the first proposed value is written on $x$, which becomes immutable from this point on \footnote{
    This definition is close to the \emph{sticky bit} and is not exactly the same as commonly accepted 
    definitions of consensus. However, it is easy to implement it using a regular consensus task and 
    one read/write register, initialized to $\bot$ and written when the consensus is decided. We use 
    this definition for the sake of clarity of the proposed algorithms. 
  }.
\end{description}
  
An execution in the infinite arrival model is a (finite or infinite) sequence of steps. In each step, one process
executes either a local transition of its algorithm, accesses an operation available on a register, or returns a value.
A process that has already returned a value cannot execute a step afterwards.
Processes are asynchronous, in the sense that there is no constraint on which process takes each step of an execution:
a process may take an unbounded but finite number of consecutive steps, or wait an unbounded but finite number of other processes' steps
between two of its own steps. Moreover, it is possible that a process stops taking steps at some point 
in the execution even if it has not returned yet, which is similar to a crash in the classical model.

We say that a process $p_i$ \emph{arrives} in an execution at the time of its first step during this execution.
Remark that, although the number of processes in the system is infinite, the number of processes that have arrived
into the system at any step is finite. 

Processes are anonymous, in the sense that they do not know their unique identifier: 
this hypothesis models the fact that it would be unreasonably costly to base an algorithm 
on unbounded identifiers, especially when the attribution of identifiers is not controlled by
the arrival order of processes in the system. Moreover, by assuming anonymous processes, the result becomes more general.

\section{The Weak Log Abstraction} 
\label{section:weakLog}
As explained previously, the main difficulty in building a wait-free universal construction in the infinite arrival model resides in replacing helping mechanisms already in place in algorithms for the finite arrival model. In this section, we introduce the weak log, an abstraction that allows each process to announce its own value. 

In the infinite arrival model, long-lived linearizable objects can be expressed as distributed tasks, as a process that invokes several operations on a long-lived object can be modelled as several processes each issuing a unique operation~\cite{castaneda:hal-01660646}. We now define the weak log abstraction as a distributed task. 

In an instance of the weak log, each process $p_i$ proposes a value through an operation $\Append(v_i)$, that returns the sequence of all the values previously appended. The weak log is wait-free 
but not linearizable, 
in the sense that there might be no inclusion between the sequences returned by different processes. 
Instead, it is specified that the value proposed by a correct process will eventually appear in all subsequent sequences.
Moreover, the order of the values inside sequences is consistent over the different sequences returned by all processes. 

\begin{definition}[weak log]
  All processes $p_0, p_1, \dots$ propose distinct values $v_0, v_1, \dots$ by invoking $\Append(v_i)$,
  that returns a finite sequence $w_i = w_{i,1} \cdot w_{i,2} \cdots w_{i,|w_i|}$ such that:

  \begin{description}
  \item[Validity] All values in a sequence $w_i$ have been appended by some process: 
    $\forall j, \exists k : w_{i,j} = v_k.$
  \item[suffixing] If process $p_i$ terminates its invocation, then its value is appended at the end of its returned sequence:
    $w_{i,|w_i|} = v_i$ 
    \item[Total order] If two processes $p_i$ and $p_j$ terminate their invocations,
    then all pairs of values that $w_i$ and $w_j$ both contain appear in the same order:
    there is no $k_i, k_j, l_i, l_j$ such that $w_{i, k_i} = w_{j, k_j}$, $w_{i, l_i} = w_{j, l_j}$, $k_i \le k_j$ and $l_i > l_j$.
  \item[Eventual visibility] If some process $p_i$ terminates its invocation, then, eventually,
    all processes terminating their invocation will return a sequence containing $v_i$. In other words, the number of
    returned sequences that do not contain $v_i$ is finite.
  \item[Wait-freedom :] No process takes an infinite number of steps in an execution.
  \end{description}

\end{definition}

\section{A Universal Construction}
\label{section:universal}

Algorithm \ref{ALGO:waitfreelogCons} presents a universal construction using a weak log and consensus objects.
This algorithm is similar to the one presented in \cite{HS08}, except that the array of
single-writer/multiple-reader registers used by processes to announce their operations
is replaced by a weak log. The shared object to implement is represented by an initial state
\InitialState and a set of operations that change the state of the object and return a value. 

Processes share two variables: 
\begin{itemize}
\item \Announcements is a weak log in which processes append their own invocation;
\item \Operations is a consensus object at the tail of a linked list of operations. The list is a succession of nodes 
of the form $\langle v, \Cons\rangle$, where $v$ is the invocation of a process and $\Cons$ is a consensus object, 
referencing another node after the consensus has been won by some process. 
\end{itemize}

When process $p_i$ calls $\Apply(\Invok)$, it first appends $\Invok$ to \Announcements and obtains a list $\ToHelp_i$ of 
invocations in return. Then, it attempts to insert invocations of $\ToHelp_i$ at the end of the list \Operations
until all the invocations of \Announcements have been inserted. While traversing the list, it maintains a state $\State_i$ of the implemented object, initialized to $\InitialState$ and on which all invocations are applied in their order of appearance in the list. 

\begin{algorithm}%[ht]
    \SubAlgo{\Operation $\Apply(\Invok)$ \Is}{
        $\ToHelp_i \gets \Announcements.\Append(\Invok)$\;
        $\Cons_i \leftarrow \Operations$;
        $\State_i \leftarrow \InitialState$\;
        \While{$\ToHelp_i\neq \varepsilon$}{
                $\langle \Winner_i, \Cons_i\rangle \leftarrow \Cons_i.\Propose(\langle \ToHelp_i[0], \bot \rangle)$\label{algo:waitfreelogCons:line:propose}\;
                $\ToHelp_i \leftarrow \ToHelp_i \setminus \Winner_i$\label{algo:waitfreelogCons:line:remove}\;
                \lIf{$\Winner_i=\Invok$}{
                    $\Result_i \leftarrow \State_i.\Invoke(\Winner_i)$\label{algo:waitfreelogCons:line:exec1}%\;
                }\lElse{
                $\State_i.\Invoke(\Winner_i)$\label{algo:waitfreelogCons:line:exec2}%
                }
            }
        \Return $\Result_i$\;
    }
  \caption{Wait-free universal construction in the infinite arrival model}
  \label{ALGO:waitfreelogCons}
\end{algorithm}

We now prove that Algorithm \ref{ALGO:waitfreelogCons} is linearizable and wait-free.
Linearizability is achieved by Algorithm \ref{ALGO:waitfreelogCons} in the same way as in 
\cite{HS08}, so the proof of Lemma~\ref{Lemma:universalLinearizable} is similar.

\begin{lemma}[Linearizability]
    \label{Lemma:universalLinearizable}
    All executions admissible by Algorithm~\ref{ALGO:waitfreelogCons} are linearizable.
\end{lemma}
\begin{proof}
  Let $\alpha$ be an execution admissible by Algorithm~\ref{ALGO:waitfreelogCons}.

  Let us first remark that, for any operation $\Apply(\Invok)$
  invoked by process $p_i$, at most one memory location $\Cons$ is such that
  $\Cons.\Head = \Invok$. Indeed, suppose this is not the case, and let us
  consider the first two such memory locations, $\Cons_j$ and $\Cons_k$.
  Both were proposed
  on line~\ref{algo:waitfreelogCons:line:propose} by some processes $p_j$ and $p_k$. 
  As operations are totally ordered in a list, Process $p_k$ accessed $\Cons_j$
  before accessing $\Cons_k$. After accessing $\Cons_j$, $\Invok \notin \ToHelp_k$ 
  because of line~\ref{algo:waitfreelogCons:line:remove}, which contradicts
  the fact that $p_k$ proposed $\Invok$ on $\Cons_k$.
  
  Let us define the linearization point of any operation $\Apply(\Invok)$
  for which there is a memory location $\Cons$ such that $\Cons.\Head = \Invok$
  as the first step in which some process invoked $\Cons.\Propose$.

  We now prove that any operation $\Apply(\Invok)$ done by a terminating process
  $p_i$ has a linearization point, between its invocation and termination point.
  By the \emph{validity} property of $\Announcements$, and as all $\Invok$ values
  are different, no process proposes $\Invok$ before $p_i$ arrived in the system.
  By the \emph{suffixing} property
  of $\Announcements$, at the beginning of $p_i$'s loop, $\Invok \in \ToHelp_i$.
  When $p_i$ terminates, $\Invok \notin \ToHelp_i$. Therefore, $\Invok$ was removed
  on line~\ref{algo:waitfreelogCons:line:remove} of some iteration of the loop,
  and $\Cons$ at the beginning of the iteration is such that $\Cons.\Head = \Invok$
  at the end of the loop.

  Finally, operations are applied by $p_i$ on $s$ in the same order as they
  appear in the list (Lines \ref{algo:waitfreelogCons:line:exec1} and
  \ref{algo:waitfreelogCons:line:exec2}),
  which is the same order as their linearization points, which concludes the proof.
\end{proof}

The proof of wait-freedom (Lemma~\ref{Lemma:universalWaitfree}) is more challenging
because the proof of \cite{HS08} heavily relies on the fact that the number of processes is finite.

\begin{lemma}[Wait-freedom]
    \label{Lemma:universalWaitfree}
    All executions admissible by Algorithm~\ref{ALGO:waitfreelogCons} are wait-free.
\end{lemma}
\begin{proof}
  Suppose there is an execution $\alpha$ admissible by Algorithm~\ref{ALGO:waitfreelogCons} that is not wait-free.
  It means that some process $p_i$ takes an infinite number of steps in $\alpha$. By the \emph{wait-freedom} property
  of $\Announcements$, $p_i$ enters the while loop after a finite number of steps, and each iteration of the
  loop terminates. Therefore, $p_i$ executes an infinite number of loop iterations.
  Let $w_i = w_{i,1} \cdot w_{i,2} \cdots w_{i,|w_i|}$ be the initial value of $\ToHelp_i$.
  As $w_i$ is finite and $p_i$ proposes some $\langle w_{i,k}, \bot\rangle$ at each iteration,
  there exists a value $k$ such that $p_i$ proposes $\langle w_{i,k}, \bot\rangle$ an infinite number of times.

  Let $win_0, win_1, \dots$ be the infinite sequence of the values taken by $\Winner_i$ during the execution,
  let $p_{\omega(j)}$ be the process that took the step on line~\ref{algo:waitfreelogCons:line:propose} installing
  the value $win_j$ in the consensus and let $p_{a(j)}$ be the process that invoked $\Apply(win_j)$.

  As processes $w_{\omega(j)}$ always propose the first invocation of $\ToHelp_{\omega(j)}$ that was not inserted in the list yet, 
  there is an infinite number of values $win_j$ such that either
  \begin{itemize}
  \item $w_{i,k}$ is not part of $w_{\omega(j)}$ or
  \item $w_{i,k}$ is part of $w_{\omega(j)}$, but appears after $win_j$ in the list. 
  \end{itemize}
  By the \emph{eventual visibility} property, the first case only concerns a finite number
  of $win_j$, so there is a finite number of values $win_j$ in the second case.

  For each of them, by the \emph{suffixing} property, $win_j = w_{a(j), |w_{a(j)}|}$,
  that is the process that invoked $\Apply(win_j)$ obtained $win_j$ at the last value of its $\ToHelp_{a(j)}$.
  By the \emph{total order} property, it is impossible that $p_{\omega(j)}$ obtains $win_j$ before $w_{i,k}$
  and $p_{a(j)}$ obtains $w_{i,k}$ before $win_j$. Therefore $w_{a(j)}$ does not contain $w_{i,k}$.
  However, this contradicts the \emph{eventual visibility} property that prevents an infinite number of $p_{a(j)}$
  processes to ignore $w_{i,k}$.

  This is a contradiction meaning the assumption of a non-wait-free execution is absurd. 
\end{proof}

\section{From Consensus to the Weak Log}
\label{section:consensus}

The main difficulty in the implementation of a weak log lies in the allocation of one memory location
per process, where it can safely announce its invoked operation. 
As it is impossible to allocate an infinite array at once, it is necessary to
build a data structure in which processes allocate their own piece of memory, and make it reachable 
to other processes, by winning a consensus. The list of Algorithm~\ref{ALGO:waitfreelogCons}
follows a similar pattern, but it poses a challenge: as an infinite number of processes access the
same sequence of consensus objects, one process may loose all its attempts to insert its own
node, breaking wait-freedom.

\begin{algorithm}[t]
  \SubAlgo{\Operation $\Append(v)$ \Is}{
    \tcp{add $v$ to the log}
    $\Node_i \gets \Last.\Read().\Propose(\langle\langle v,\bot \rangle,\bot\rangle)$\label{ALGO:wfclCons:line:propose}\;
    $\Last.\Write(\Node_i.\Tail)$\label{ALGO:wfclCons:line:writelast}\;
    \While(\label{ALGO:wfclCons:line:secondaryloop}){$\Node_i.\Head \neq v$}{
      $\Node_i \gets \Node_i.\Tail.\Propose(\langle v, \bot \rangle)$\label{ALGO:wfclCons:line:secondarypropose}\;
    }
    \tcp{read the log}
    $\Log_i \gets \varepsilon$;
    $\List_i \gets \First$;
    $\Node_i \leftarrow \List_i.\Head$\;
    \While(\label{ALGO:wfclCons:line:collectloop}){\True}{
      $\Log_i \leftarrow \Log_i \concat \Node_i.\Head$\label{ALGO:wfclCons:line:collectconcat}\;
      \lIf{$\Node_i.\Head = v$}{\Return $\Log_i$\label{ALGO:wfclCons:line:return}}
      $\Node_i \leftarrow \Node_i.\Tail$\label{ALGO:wfclCons:line:nodechange}\;
      \If{$\Node_i = \bot$}{
        $\List_i \gets \List_i.\Tail$\;
        $\Node_i \leftarrow \List_i.\Head$\label{ALGO:wfclCons:line:listchange};
      }
    }
  }
  \caption{Wait-free weak log using consensus}
  \label{ALGO:wfclCons}
\end{algorithm}

Algorithm~\ref{ALGO:wfclCons} solves this issue by using a novel feature,
that we call \emph{passive helping}: when a process wins a consensus,
it creates a side list to host values of processes concurrently competing
on the same consensus object. As only a finite number of processes have
arrived in the system when the consensus is won, a finite number of
processes will try to insert their value in the side list, which
ensures termination. Figure~\ref{figure:wfclCons} presents an
execution of Algorithm~\ref{ALGO:wfclCons}.

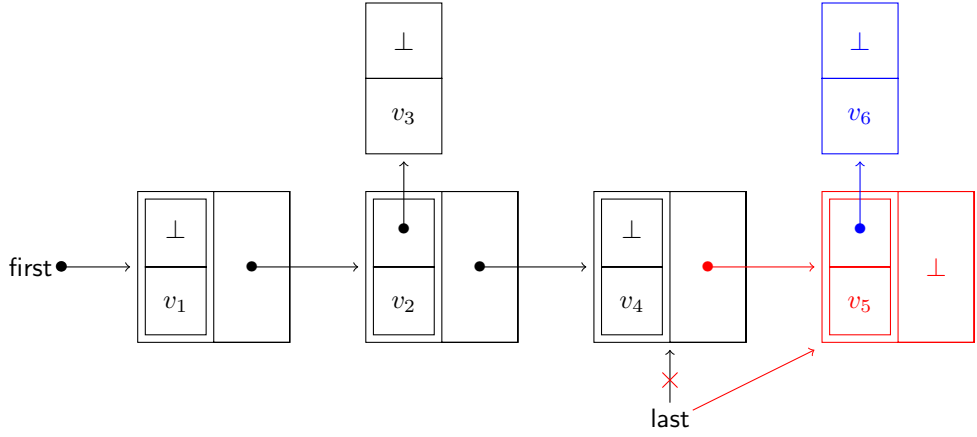
\begin{figure}[t]
  \centering
  \begin{tikzpicture}
  \draw (7,-1) node{\Last};
  \draw (-1,1) node[left]{$\First$};
  \draw (-1,1) node{$\bullet$};  

  \draw (0,0) rectangle +(2,2) rectangle +(1,0) +(0.1,0.1) rectangle +(0.9,1) rectangle +(0.1,1.9) +(0.5,0.5) node{$v_1$} +(0.5,1.5) node{$\bot$} ;
  \draw (3,0) rectangle +(2,2) rectangle +(1,0) +(0.1,0.1) rectangle +(0.9,1) rectangle +(0.1,1.9) +(0.5,0.5) node{$v_2$} +(0.5,1.5) ;
  \draw (6,0) rectangle +(2,2) rectangle +(1,0) +(0.1,0.1) rectangle +(0.9,1) rectangle +(0.1,1.9) +(0.5,0.5) node{$v_4$} +(0.5,1.5) node{$\bot$} ;
  \draw[red] (9,0) rectangle +(2,2) rectangle +(1,0) +(0.1,0.1) rectangle +(0.9,1) rectangle +(0.1,1.9) +(0.5,0.5) node{$v_5$} +(0.5,1.5) +(1.5,1) node{$\bot$};

  \draw (3,2.5) rectangle +(1,1) rectangle +(0,2) +(0.5,0.5) node{$v_3$} +(0.5,1.5) node{$\bot$};
  \draw[blue] (9,2.5) rectangle +(1,1) rectangle +(0,2) +(0.5,0.5) node{$v_6$} +(0.5,1.5) node{$\bot$};

  \draw[->] (7,-0.8) -- (7,-0.1);
  \draw[red, ->] (7.3,-0.9)  -- (8.9,-0.1);

  \draw[->] (-1,1) node{$\bullet$} -- (-0.1,1);
  \draw[->] (0,0) +(1.5,1) node{$\bullet$} -- +(2.9,1);
  \draw[->] (3,0) +(1.5,1) node{$\bullet$} -- +(2.9,1);
  \draw[red,->] (6,0) +(1.5,1) node{$\bullet$} -- +(2.9,1);

  \draw[->] (3,0) +(0.5,1.5) node{$\bullet$} -- +(0.5,2.4);
  \draw[blue,->] (9,0) +(0.5,1.5) node{$\bullet$} -- +(0.5,2.4);
  \draw[red] (7,-0.5) +(-0.1,-0.1) -- +(0.1,0.1) +(-0.1,0.1) -- +(0.1,-0.1);
\end{tikzpicture}
\caption{In this example,
  $p_5$ inserting $v_5$ and $p_6$ inserting $v_6$
  both read the same value
  $\langle\langle v_4, \bot \rangle, \bot \rangle$ of \textsf{last}.
  Process $p_5$ wins the consensus and inserts $v_5$ after $v_4$,
  and $p_6$ looses the consensus, so, it inserts $v_6$
  in the side list created by $p_5$.
}
\label{figure:wfclCons}
\end{figure}

Processes executing Algorithm~\ref{ALGO:wfclCons} share two variables: $\First$ and $\Last$ defined as follows.
\begin{itemize}
\item $\First$ is a consensus object that accepts values of the form 
$\List = \langle\List.\Head, \List.\Tail\rangle$ to be proposed,
where $\List.\Tail$
is a consensus object that stores values of the same type as $\First$.
$\List.\Head$ is a node of the side list of the form
$\Node = \langle \Node.\Head, \Node.\Tail\rangle$, where
$\Node.\Head$ is a value appended by some process and
$\Node.\Tail$ is a consensus object containing
values of the same type as $\Node$.
In other words, $\First$ references
the beginning of a list of lists of appended values. 
\item $\Last$ is a read/write register referencing a consensus object 
of the same type as $\First$, and initialized to the address of $\First$. 
In absence of concurrency, $\Last$ references the end of the list starting with $\First$. 
\end{itemize}

When process $p_i$ invokes $\Append(v)$, it first reads $\Last$, then
it proposes a list containing $v$ as its successor and it writes
the value returned by the consensus in $\Last$.
If $p_i$ loses the consensus, it inserts $v$ in the side list of the successor
(Lines~\ref{ALGO:wfclCons:line:secondaryloop}~and~\ref{ALGO:wfclCons:line:secondarypropose}).
After that, $p_i$ traverses the list of lists to build the sequence $\Log_i$ it returns ($\concat$ represents concatenation). 

Note that the consensus and the write on lines~\ref{ALGO:wfclCons:line:propose} and~\ref{ALGO:wfclCons:line:writelast}
are not done atomically. This means that a very old value can be written in $\Last$, in which case its value could move backward. The central property of the algorithm, proved by Lemma~\ref{LEMMA:lastprogress}, is that $\Last$ eventually moves forward, allowing very slow processes to find some place in a side list.

\begin{lemma}
    If an infinite number of processes execute line~\ref{ALGO:wfclCons:line:writelast}, then
    the number of processes that read the same \Last value at line \ref{ALGO:wfclCons:line:propose} is finite.
    \label{LEMMA:lastprogress}
\end{lemma}
\begin{proof}
We first prove by induction on the list of lists of the weak log that for all $\List = \langle\List.\Head, \List.\Tail\rangle$, the number of \Write operations of $\List$ in \Last at line \ref{ALGO:wfclCons:line:writelast} is finite.
\begin{itemize}
    \item Initially, \First is never written in \Last, because only decided values on line \ref{ALGO:wfclCons:line:propose} are written, and \First is never proposed.
    \item We now prove that, if the number of writes of $\List$ in \Last is finite, 
    then the number of writes of $\List.\Tail$ in \Last is finite.
    
    We prove the following contrapositive proposition: if the number of writes of $\List.\Tail$ in \Last is infinite, then the number of writes of $\List$ in \Last is infinite as well.
    
    In order to write $\List.\Tail$ in \Last, a process needs to read \List in \Last at line \ref{ALGO:wfclCons:line:propose}. As $\List.\Tail$ is written an infinite number of times and 
    \List is read an infinite number of times, then necessarily, \List is written an infinite number 
    of times as well.
    \end{itemize}
%So we have the following property: for all \List $l$ in the weak log, the number of \Write operations of $l$ in \Last at line \ref{ALGO:wfclCons:line:writelast} is finite.
    
%We now prove by contradiction that the number of read operations of \Last that return $\List$ is finite:
Let us now suppose that an infinite number of processes execute line~\ref{ALGO:wfclCons:line:writelast}, 
and that an infinite number of reads of \Last return $\List$.
This implies that there was an infinite number of \Write operations of $\List$ in \Last at line \ref{ALGO:wfclCons:line:writelast}, which contradicts the previous induction result.
\end{proof}

\begin{lemma}[Validity]
  All values in the sequence $w_i$ have been appended by some process.
\end{lemma}
\begin{proof}
    The value $\Log$ returned by the algorithm, is built by concatenation of values 
    $\Node.\Head$, that can only be created on line 
    \ref{ALGO:wfclCons:line:propose} or \ref{ALGO:wfclCons:line:secondarypropose},
    using an appended value.
\end{proof}

\begin{lemma}[suffixing]
  $w_{i,|w_i|} = v_i$ 
\end{lemma}
\begin{proof}
  This is a direct consequence of the fact that $v_i$ is appended
  at the end of $\Log_i$ at line \ref{ALGO:wfclCons:line:collectconcat} just before the return statement on line  \ref{ALGO:wfclCons:line:return}.
\end{proof}

Definition~\ref{def:precedence} formalizes the order in which values are ordered in the weak log.
Intuitively, this order is the concatenation of all the side lists.
In Algorithm~\ref{ALGO:wfclCons}, the list of list is traversed in this precedence order, which ensures consistency of the order of all returned sequence (\emph{Total order} property of the weak log). 

\begin{definition}[Precedence]\label{def:precedence}
  A value $v_i$ precedes a value $v_j$ in the weak log if:
  \begin{itemize}
      \item $v_i$ is in a \List $l_i$ and $v_j$ is in a \List $l_j$ such that there exists a sequence of lists
      $\{l_1,\dots , l_n\}$ such that for all $l_k$, $l_k = \langle node_k,l_{k+1} \rangle$, 
      and $l_i = \langle node_i, l_1\rangle \wedge l_n = \langle node_n , l_j\rangle$.
      \item $v_i$ and $v_j$ are in the same \List $l$ and there exists a sequence of nodes $\{n_1, \dots , n_n \}$
      such that for all $n_k$, $n_k = \langle v,n_{k+1}\rangle$, $n_1=\langle v_i,n_2\rangle$ and $n_n = \langle v_j, n_{n+1}\rangle$.
  \end{itemize}
  
  We extend the notion of precedence to lists and nodes using the same definition.
\end{definition}

\begin{lemma}[Total order]
  If two processes $p_i$ and $p_j$ terminate their invocations,
  then all pairs of values that $w_i$ and $w_j$ both contain appear in the same order.
\end{lemma}
\begin{proof}
    Let us remark that both processes $p_i$ and $p_j$ append values in their log following the precedence order
    defined by Definition~\ref{def:precedence}. Therefore, for any two values $v_k$ and $v_l$ that are both contained by $w_i$ and $w_j$,
    $p_i$ and $p_j$ have appended them at the end of the log in the same order, which proves the lemma.
\end{proof}

\begin{lemma}[Eventual visibility]
  If some process $p_i$ terminates its invocation, then the number of
  returned sequences that do not contain $v_i$ is finite.
\end{lemma}
\begin{proof}
  Let us denote by $l_i$ and $n_i$ the values of $\List_i$ and $\Node_i$ when $p_i$ terminates.
  
  Let us suppose (by contradiction) that there is an infinite number of processes whose returned sequence
  does not contain $v_i$, and an infinite number of them started their operation after $p_i$ returned.
  For each process $p_j$ of them, let $l_j$ and $n_j$ be the values of $\List_j$ and $\Node_j$ when $p_j$
  terminates its execution. As the collect loop respects the precedence order of Definition~\ref{def:precedence},
  for an infinite number of $p_j$, $l_j$ precedes $l_i$.
  As there is only a finite number of lists preceding $l_i$ ($p_i$ terminates), an infinite number of processes
  have the same value $l$ of $l_j$. All of them read the same value of \Last at line \ref{ALGO:wfclCons:line:propose} and wrote on line \ref{ALGO:wfclCons:line:writelast}.
  This contradicts Lemma \ref{LEMMA:lastprogress}.
\end{proof}

\begin{lemma}[Wait-freedom]
  No process takes an infinite number of steps in an execution.
\end{lemma}
\begin{proof}

    Let us suppose that there exists an execution $\alpha$ such that process $p_i$ takes an infinite number of steps in $\alpha$ trying to append $v_i$.
    This means that either of the two loops (lines \ref{ALGO:wfclCons:line:secondaryloop} and \ref{ALGO:wfclCons:line:collectloop}) are looping an infinite number of times:
    \begin{itemize}
        \item If the loop at line \ref{ALGO:wfclCons:line:secondaryloop} loops for an infinite number of times, it means that $\Node.\Head \neq v_i$ for an infinite number of nodes.
        This implies that an infinite number of values are appended in the same list at line \ref{ALGO:wfclCons:line:secondarypropose}, which means that an infinite number of processes 
        read the same value at line \ref{ALGO:wfclCons:line:propose}, and wrote at line \ref{ALGO:wfclCons:line:writelast} which contradicts Lemma \ref{LEMMA:lastprogress}.
        \item If the loop at line \ref{ALGO:wfclCons:line:collectloop} loops for an infinite number of times, this means that $p_i$ never reads $v_i$, and as there is a finite number of lists that precedes the list $l_i$ in which $v_i$ has been appended, one of these lists contains an infinite number of nodes. All these nodes were created by processes reading the same value of \Last in line \ref{ALGO:wfclCons:line:propose}, which contradicts Lemma \ref{LEMMA:lastprogress}.
    \end{itemize}
    Both cases lead to a contradiction, so the algorithm is wait-free.
\end{proof}

\section{From Compare\&Swap to the Weak Log}
\label{section:CAS}

In multi-threaded environments, consensus is usually replaced by special operations like
compare\&swap. In this section, we replace read/write registers and consensus objects
by CAS registers. A CAS register $x$ provides two operations: a read operation $x$ working as the 
read operation of the immutable register and a compare\&swap operation $x.\CAS(\mathit{expect}, \mathit{update})$ that atomically checks if the
current value of $x$ is $\mathit{expect}$, changes it to $\mathit{update}$ and returns \True if it is the case.
If $x\neq \mathit{expect}$, then $x.\CAS(\mathit{expect}, \mathit{update})$ 
returns \False without changing the value of $x$.
On the one hand, like in the finite arrival model consensus objects and CAS registers are computationally equivalent: 
a consensus object can be easily simulated by
implementing $x.\Propose(v)$ as $x.\CAS(\bot, v)$, and conversely, compare\&swap can be
implemented using consensus as proved in the previous section. On the other hand,
compare\&swap is more flexible as it allows to write several times in the same register.

\begin{algorithm}[t]
  \SubAlgo{\Operation $\Append(v)$ \Is}{
    \tcp{add $v$ to the log}
    $\Next_i \leftarrow \Last$\;
    \If{$\lnot\Last.\CAS(\Next_i, \langle v, \Next_i\rangle)$}{
      \Repeat {$\First_i.\Tail.\CAS(\Next_i, \langle v, \Next_i\rangle)$}{
        $\First_i \leftarrow \Next_i$\;
        $\Next_i \leftarrow \First_i.\Tail$\;
      }
    }
    \tcp{read the log}
    $\Log_i \leftarrow v$\;
    \While{$\Next_i \neq \Empty$}{
        $\Log_i \leftarrow \Next_i.\Head \concat \Log_i$\;
        $\Next_i \leftarrow \Next_i.\Tail$\;
    }
    \Return $\Log_i$\;
  }
  
  \caption{Wait-free weak log using compare\&swap}
  \label{ALGO:wfclCas}
\end{algorithm}

\begin{figure}[t]
  \centering
\begin{tikzpicture}
  \draw (4,-0.8) node{\Last};
  \draw[blue] (1.5,1.5) rectangle +(1,1) rectangle +(2,0) +(0.5,0.5) node{$v_5$} +(1.5,0.5) node{$\bullet$};
  \draw[red] (0,0) rectangle +(1,1) rectangle +(2,0) +(0.5,0.5) node{$v_4$} +(1.5,0.5) node{$\bullet$};
  \draw (3,0) rectangle +(1,1) rectangle +(2,0) +(0.5,0.5) node{$v_3$} +(1.5,0.5) node{$\bullet$};
  \draw (6,0) rectangle +(1,1) rectangle +(2,0) +(0.5,0.5) node{$v_2$} +(1.5,0.5) node{$\bullet$};
  \draw (9,0) rectangle +(1,1) rectangle +(2,0) +(0.5,0.5) node{$v_1$} +(1.5,0.5) node{$\bullet$};
  \draw (12.5,0.5) node{\Empty};

  \draw[->] (4,-0.8) -- (4,-0.1);
  \draw[red, ->] (3.6,-0.75) -- (1.1,-0.1);

  \draw[red,->] (0,0) +(1.5,0.5) -- +(2.9,0.5);
  \draw[->] (3,0) +(1.5,0.5) -- +(2.9,0.5);
  \draw[->] (6,0) +(1.5,0.5) -- +(2.9,0.5);
  \draw[->] (9,0) +(1.5,0.5) -- +(2.9,0.5);

  \draw[blue, ->] (0,0) +(1.5,0.5) -- +(2.4,1.4);
  \draw[blue, ->] (1.5,1.5) +(1.5,0.5) -- (3.9,1.1);

  \draw[blue] (2.5,0.5) +(-0.1,-0.1) -- +(0.1,0.1) +(-0.1,0.1) -- +(0.1,-0.1);
  \draw[red] (4,-0.4) +(-0.1,-0.1) -- +(0.1,0.1) +(-0.1,0.1) -- +(0.1,-0.1);
\end{tikzpicture}
\caption{In this example, $p_4$ inserting $v_4$ and $p_5$ inserting $v_5$ both read the same value $v_3$ of \textsf{last}.
  Process $p_4$ wins its compare\&swap and inserts $v_4$ before $v_3$, and $p_5$ looses the compare\&swap to $p_4$ and inserts
  $v_5$ after $v_4$.}
\label{figure:wfclCas}
\end{figure}
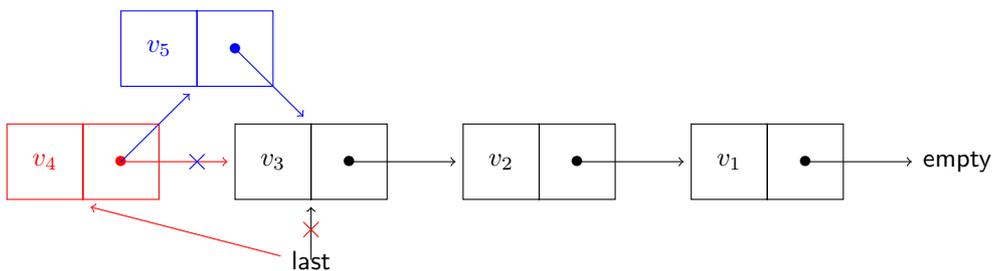

Algorithm \ref{ALGO:wfclCas} presents a simplified implementation of the weak log based on
CAS registers instead of consensus objects. There are two main differences between Algorithm
\ref{ALGO:wfclCas} and Algorithm \ref{ALGO:wfclCons}. First, values are stored in a list
instead of a list of lists. Second, the list is managed as a stack, and not as a queue.
Figure~\ref{figure:wfclCas} illustrates an execution of the algorithm. 

Processes share a unique CAS register $\Last$, that stores either the initial value, $\Empty$,
or nodes of the form $\langle \Head, \Tail \rangle$, where $\Head$ is a value appended by
a process and $\Tail$ has the same type as $\Last$. 

When a process invokes $\Append(v)$, it first attempts to add $v$ at the top of the stack,
by winning a compare\&swap on $\Last$. If this attempt fails and a value $v'$ was
inserted concurrently by $p_j$, $p_i$ continuously tries to inserts $v$ after $v'$. Similarly
to Algorithm \ref{ALGO:wfclCons}, $p_i$ eventually succeeds because it only competes
in this task with the finite set of processes that read $\Last$ before $p_j$ won its
compare\&swap.
After $p_i$ inserted $v$ in the list, it traverses the list to its end to build its return sequence. 
Traversal is done in a last-in-first out manner, so read values are appended at the beginning of $\Log_i$. 

\section{Conclusion}
\label{section:conclusion}

Consensus is a central problem in distributed computing,
because it allows wait-free linearizable implementations of all objects
with a sequential specification, in systems composed of $n$ asynchronous processes
that may crash. 
In this paper, we asked the question of whether the result still hold
in the infinite arrival model, in which a potentially infinite number of processes 
can arrive and leave during an execution. We answered this question positively by
giving a wait-free and linearizable universal construction using only
consensus objects and read/write registers.

Our proposed construction is based on two lists containing all the
operations invoked on the object. An interesting question is whether
it is possible to reduce the cost in memory by eventually removing
the operations from both lists. Our approach does not allow
such optimizations, as it relies on a ``passive helping'' mechanism,
in which processes winning a consensus instance allocate memory locations
to host the value of potential processes that loose
all the consensus instances they are involved in.

\section{Acknowledgements}
This work was partially supported by the French ANR project 16-CE25-0005 O'Browser
devoted to the study of decentralized applications on Web browsers.

\bibliography{infiniteCN}

\end{document}